\documentclass[a4paper,11pt]{amsart}
\usepackage{amsthm}
\usepackage{amsmath}
\usepackage{geometry}
\usepackage{amsfonts}
\usepackage{graphicx}
\usepackage{array}
\usepackage{amssymb}
\usepackage{upgreek}
\usepackage{mathrsfs}
\usepackage{float}
\usepackage{hyperref}
\usepackage{enumitem}

\usepackage{pgf}
\usepackage{tikz}
\usetikzlibrary{arrows}
\usetikzlibrary{matrix}
\usepackage{color}
\usepackage{subcaption}
\newtheorem{theorem}{Theorem}
\newtheorem{corollary}[theorem]{Corollary}
\newtheorem{proposition}[theorem]{Proposition}
\newtheorem{definition}[theorem]{Definition}
\newtheorem{lemma}[theorem]{Lemma}
\newtheorem{example}[theorem]{Example}

\newtheorem{remark}[theorem]{Remark}

\numberwithin{equation}{section}   % Para numerar ecuaciones por secciones.
\begin{document}

\title[Causal decrease and null lines]{\textbf{Remarks on causal decrease, omniscient foliations and null lines in spacetimes}}

%%% Para \documentclass{amsart}
%%%%%%%%%%%%%%%%%%%%%%%%%%%%%%%%
\author{A. Bautista}
\address{Dpto. de An\'{a}lisis Econ\'{o}mico: Econom\'{i}a Cuantitativa, Univ. Aut\'{o}noma de Madrid, C/ Francisco Tom\'{a}s y Valiente, 5, 28049 Madrid, Spain.}
\email{alfredo.bautista@uam.es}
%\date{}
%\subjclass{Primary 05C38, 15A15; Secondary 05A15, 15A18, ???}

\begin{abstract}
This work provides a counterexample to the previously asserted non--existence of future null lines in strongly causal spacetimes $(M, \mathbf{g})$ foliated by a past--omniscient timelike foliation $\mathcal{F}$ with a compact leaf space $Q$. 
We introduce a new hypothesis requiring the timelike vector field that defines $\mathcal{F}$ to be causally decreasing. Under this condition, the causality requirement on $M$ can be relaxed to the distinguishing level. 
In the absence of this hypothesis, we demonstrate that more restrictive conditions, such as global hyperbolicity, are necessary to ensure the absence of future null lines. 
Finally, we derive sufficient conditions on the scale function of Generalized Robertson--Walker (GRW) spacetimes to ensure that the natural timelike vector field is causally decreasing, thereby guaranteeing the non--existence of future null lines.

\end{abstract}

\maketitle

%\tableofcontents

\section{Introduction}

In \cite{Ba88}, Bartnik introduced his well--known conjecture on cosmological spacetimes: a globally hyperbolic spacetime $M$ with compact Cauchy surfaces that satisfies the timelike convergence condition, that is $\mathrm{Ric}(V,V) \geq 0$ for any timelike $V \in TM$, must either be timelike geodesically incomplete or split as $M = \mathbb{R} \times S$ with the product metric $\mathbf{g} = -dt^2 \oplus \mathbf{h}$, where $(S, \mathbf{h})$ is a compact Riemannian manifold.

Several attempts have been made to verify this conjecture. Notable contributions by Galloway \cite{Ga84, Ga89, Ga97} introduced the non--existence of null lines as an additional hypothesis. Consequently, the presence or absence of achronal null geodesics has become an important property in determining whether $M$ admits a metric splitting.

Later, Harris and Low \cite{HL01} studied chronological spacetimes foliated by timelike curves to determine how the leaf space $Q$ represents the \emph{shape of space}. They showed that each timelike foliation $\mathcal{F}$ is defined by a complete timelike vector field \cite[Lem. 1.1]{HL01}. Furthermore, if there are no \emph{ancestral pairs} in $M$, that is there are no $\gamma_1$, $\gamma_2$ leaves of $\mathcal{F}$ such that $p_1 \subset I^{+}(p_2)$ for all $p_i \in \gamma_i$ with $i=1,2$, then $Q$ is Hausdorff and $M$ is diffeomorphic to the product $\mathbb{R} \times Q$ \cite[Thm. 1.2]{HL01}. It is important to note that this provides a topological decomposition rather than a metric splitting in Bartnik’s sense, as the metric is not necessarily a product.

In \cite{HL01}, the concepts of omniscient foliation and causal monotonicity were also introduced. The former implies that every leaf $\gamma$ of the timelike foliation $\mathcal{F}$ can receive information from any point in the spacetime, that means $I^{-}(\gamma) = M$ for all $\gamma \in \mathcal{F}$. The latter, described as a generalization of conformal--Killing vector fields, requires the existence of a timelike vector field whose forward flow preserves causal vectors (\emph{causal decrease}). The existence of such a field precludes ancestral pairs of leaves, thereby ensuring the product structure $M = \mathbb{R} \times Q$.

In \cite{CFH22}, the authors investigated whether a past--omniscient foliation $\mathcal{F}$ implies that the chronological past of any future--inextensible timelike curve $\lambda$ covers the entire spacetime, that is $I^{-}(\lambda) = M$. Under the assumptions of strong causality, a compact leaf space $Q$, and a past--omniscient $\mathcal{F}$, they asserted that $M$ contains no future null lines \cite[Prop. 2.1]{CFH22}.
The statement is enunciated as: 
\begin{quote}
Suppose $(M, \mathbf{g})$ is a strongly causal spacetime, and let $X \in \mathfrak{X}(M)$ be a complete timelike vector field whose flow $\Phi$ admits a compact slice $\Sigma$. Assume also that $X$ spans a future (resp. past) omniscient timelike foliation $\mathcal{F}$. Then, the future (resp. past) NOH condition holds; hence, $(M, \mathbf{g})$ has no null lines.
\end{quote}

However, the proof provided in \cite{CFH22} is flawed; it overlooks the fact that since the hypersurfaces $Q_s = \{s\} \times Q$ are not necessarily achronal, the canonical projection onto $\mathbb{R}$ of a future--inextensible causal curve may diverge toward $-\infty$ as its parameter tends toward its future limit. Consequently, the theorem as stated does not hold.

The first objective of this article is to provide a counterexample (Example \ref{ejem-counterexample}) that demonstrates the failure of the aforementioned result. This establishes the necessity of additional hypotheses to ensure the same conclusion. Our second objective is to propose these new hypotheses.
Specifically, we show that by requiring the timelike vector field defining $\mathcal{F}$ to be causally decreasing, the causality condition can be relaxed such that $M$ is distinguishing. Alternatively, without the causal decrease assumption, more restrictive conditions such as global hyperbolicity must be imposed. Under this latter condition, the statement in \cite[Prop. 2.1]{CFH22} can be recovered.

Finally, we apply our results to Generalized Robertson–Walker (GRW) spacetimes to find sufficient conditions on the warping function that guarantee causal decrease, thus ensuring the absence of null lines.

The paper is organized as follows: Section \ref{sec:prelim} provides the mathematical background on causality, timelike foliations and causal decrease. Section \ref{sec:causaldecres} presents the proof of the main theorem and its corollaries. Finally, Section \ref{sec:GRWs} discusses the application of the main result to GRW spacetimes.

\section{Preliminaries}\label{sec:prelim}

Throughout this article, we will consider $(M,\mathbf{g})$ (or simply $M$) a $n$--dimensional Lorentz manifold, $n\geq 2$.

\subsection{Causality}

Let us recall some classical definitions related to causal structure of spacetimes. 

As usual, we will denote the chronological relationships as $I^{\pm}$ and $\ll$, while $J^{\pm}$ and $\leq$ will denote the causal relationships, so for $p,q\in M$
\begin{align*}
q\in I^{+}(p) \Leftrightarrow p\ll q \quad  \text{and} \quad q\in J^{+}(p) \Leftrightarrow p\leq q & \\
q\in I^{-}(p) \Leftrightarrow q\ll p \quad \text{and} \quad q\in J^{+}(p) \Leftrightarrow q\leq p &
\end{align*}
We will also denote as $p<q$ if $p\leq q$ but $p\neq q$. 

A spacetime $M$ is said to be \emph{past--distinguishing} (resp. \emph{future--distinguishing}) if for any $p\neq q\in M$ then $I^{-}(p)\neq I^{-}(q)$ (resp. $I^{+}(p)\neq I^{+}(q)$). We will say that $M$ is \emph{distinguishing} if $M$ is both past and future distinguishing. 
This will be a fundamental hipothesis throughout this article. 

It should be noted that closed causal curves cannot exist in a distinguishing spacetime; therefore, spacetimes such as G\"{o}del’s universe, Taub-NUT and the interior of a Kerr black hole do not fall within this class and lie outside the scope of the main results of this work.

\begin{definition}
A non--void set $W\subset M$ is an \emph{irreducible past set}, or \emph{IP} for short, if 
\begin{enumerate}
\item\label{IP1} $W$ is open.
\item\label{IP2} $W$ is a past set, that is $I^{-}(W)= W$. 
\item $W$ is not the union of two proper subsets satisfying previous properties (\ref{IP1}) and (\ref{IP2}).
\end{enumerate}
We will say that an IP $W$ is a \emph{proper IP} or \emph{PIP} if $W=I^{-}(p)$ for some $p\in M$. Otherwise, we will say that $W$ is a \emph{terminal IP} or \emph{TIP}. 

Considering future sets, \emph{IF}, \emph{PIF} and \emph{TIF} are defined in an analogous way.
\end{definition}

Hereafter, for the sake of brevity, we will present only one of the time--symmetric statements since the dual statements can be obtained in an analogous manner by swapping past and future.

The following result \cite[Lem. 2.2]{GH25} is a consequence of \cite[Thms. 2.1 and 2.3]{GKP72} and it characterizes TIPs in any distinguishing spacetime. 

\begin{theorem}\label{TIP-NOHequivalence}
Let $M$ be a distinguishing spacetime and $W\subset M$ a subset. Then
\[
\begin{tabular}{c}
$W$ is a TIP \\
$\Updownarrow$ \\
$W=I^{-}(\mu)$ for a future--inextendible causal curve $\mu$ \\
$\Updownarrow$  \\
$W=I^{-}(\gamma)$ for a future--inextendible timelike curve $\gamma$
\end{tabular}
\]
\end{theorem}

Next, we rewrite \cite[Thm. 2.5]{GKP72} to align it with the aims of this article.

\begin{theorem}\label{GeneratorNullGeodesic}
Let $M$ be a distinguishing spacetime and $P\subset M$ a past--set. Every point $p\in \partial P$ is the past endpoint of an inextendible null geodesic $\mu\subset \partial P$ if and only if $P$ can be expressed as a union of TIPs.
\end{theorem}

\begin{definition}\label{def-NOHcondition}
We will say that $M$ satisfies the \emph{no--future--observer--horizon condition} (or \emph{NFOH condition}, for short) if $I^{-}(\gamma) = M$ for all future inextendible timelike curve $\gamma\subset M$.
Dually, we define the \emph{no--past--observer--horizon condition} or \emph{NPOH condition}.

We will say that $M$ holds the \emph{no--observer--horizon condition} (or \emph{NOH condition}, for short) if $M$ satisfies both NFOH and NPOH conditions.
\end{definition}

Observe that, if $\gamma:(a,b)\subset \mathbb{R}\rightarrow M$ is an inextendible future--directed timelike curve, then its restriction to any interval $[c,b)$ with $a<c<b$ satisfies $I^{-}(\gamma)=I^{-}(\gamma\vert_{[c,b)})$.  Hence, trivially one can show that the NOH condition can be defined equivalently as $I^{-}(\gamma)=I^{+}(\gamma)=M$ for every inextendible timelike curve $\gamma\subset M$.

Recall that a future--inextendible null geodesic $\mu$ is called \emph{future null line} if it is achronal, that is, $\mu\cap I^{-}(\mu)=\varnothing$. A \emph{past null line} can be defined analogously.   

In \cite[Thm. 3.1]{Ga84}, Galloway shows, under the hypothesis that $M$ is globally hyperbolic with compact Cauchy surfaces, the equivalence between the NFOH condition and the absence of future null lines. The proof of this equivalence is also true when $M$ is merely a distinguishing spacetime. We present a proof of this result following Galloway's proof in \cite{Ga84}.

\begin{proposition}\label{prop-NFOH-null-lines}
Let $M$ be a distinguishing spacetime. The NFOH condition holds if and only if there is no future null lines.
\end{proposition} 

\begin{proof}
Let us assume that there is a future null line $\mu$. Since $\mu$ is achronal, then $\mu\cap I ^{-}(\mu) = \varnothing$ but this contradicts that $I^{-}(\mu)=M$ in virtue of Theorem \ref{TIP-NOHequivalence}, since $I^{-}(\mu)=I^{-}(\gamma)=M$ for some future--inextendible timelike curve $\gamma$.

Conversely, let $\gamma$ be a future--inextendible timelike curve and let us assume that $I^{-}(\gamma)\neq M$.
Then there is $p\in \partial I^{-}(\gamma)$ and, by Theorem \ref{GeneratorNullGeodesic}, there is also a future inextendible null geodesic $\mu\subset \partial I^{-}(\gamma)$ with past endpoint at $p$. 
So, we have that $I^{-}(\mu)\subset I^{-}\left(\partial I^{-}(\gamma)\right)\subset I^{-}(\gamma)$ and then $\mu\cap I^{-}(\mu)\subset \partial I^{-}(\gamma)\cap I^{-}(\gamma)=\varnothing$ because $I^{-}(\gamma)$ is open. 
This implies that $\mu$ is achronal and hence also a future null line contradicting the hypothesis. Therefore $I^{-}(\gamma)=M$ as we claimed.
\end{proof}

\subsection{Timelike foliations}

Following \cite{HL01}, we will introduce some concepts we will use later. 

It is known \cite[Lem. 1.1]{HL01} that, in a chronological spacetime $M$, any (differentiable) timelike foliation $\mathcal{F}$ can be defined by the integral curves of a complete timelike vector field $U\in\mathfrak{X}(M)$. Then the set of leaves of $\mathcal{F}$ coincides with the space of integral curves of $U$ and it will be denoted by $Q$\footnote{For simplicity, although $Q$ depends on $\mathcal{F}$, we will not indicate the corresponding foliation, for example $Q_{\mathcal{F}}$, to denote the leaf space $Q$.}.

\begin{definition}\label{def-omnis-ancestral}
Let $U\in\mathfrak{X}(M)$ be a timelike vector field. 
\begin{enumerate}
 
\item The timelike foliation $\mathcal{F}$ defined by integral curves of the vector field $U$ is said to be \emph{past--omniscient} if $I^{-}(\gamma)=M$ for all $\gamma \in \mathcal{F}$. Dually, one can define \emph{future--omniscience}.

\item Two inextendible timelike curves $\gamma_1 , \gamma_2\subset M$ form an \emph{ancestral pair} if for all $p\in \gamma_1$ and $q\in \gamma_2$ we have $p\in I^{-}(q)$. In this case, we will say that $\gamma_1$ is \emph{ancestral} to $\gamma_2$.
\end{enumerate}
\end{definition}

Recall that, according to \cite[pg. 621]{Pe79}, a \emph{locally naked singularity} $W$ is a TIP such that $W\subset I^{-}(q)$ for some $q\in M$ or a TIF such that $W\subset I^{+}(p)$ for some $p\in M$.

Note that the existence of an ancestral pair of a timelike foliation implies the existence of locally naked singularities. Indeed, if $\gamma_1$ is ancestral to $\gamma_2$, since  $p\in I^{-}(q)$ for all $p\in \gamma_1$ and $q\in \gamma_2$ then $\gamma_1 \subset I^{-}(q)$ and $\gamma_2 \subset I^{+}(p)$ and, therefore, $I^{-}(\gamma_1) \subset I^{-}(q)$ and $I^{+}(\gamma_2) \subset I^{+}(p)$.
By theorem \ref{TIP-NOHequivalence}, since $\gamma_1$ and $\gamma_2$ are inextendible timelike curves, then $I^{-}(\gamma_1)$ is a TIP and $I^{+}(\gamma_2)$ is a TIF, both of which are locally naked singularities.

Moreover, we can state the following proposition summarizing \cite[Thm. 1.2 \& Prop. 1.3]{HL01}.

\begin{proposition}\label{prop-omnis}
Let $M$ be a chronological spacetime and $\mathcal{F}$ a timelike foliation of $M$ with leaf space $Q$. Consider the following statements:
\begin{enumerate}
\item\label{prop-omnis-1} $\mathcal{F}$ is past--omniscient.
\item\label{prop-omnis-2} $\mathcal{F}$ has no ancestral pairs.
\item\label{prop-omnis-3} $Q$ is Hausdorff and $M\simeq\mathbb{R}\times Q$ diffeomorphically.
\end{enumerate}
Then the implications (\ref{prop-omnis-1})$\Rightarrow$(\ref{prop-omnis-2})$\Rightarrow$(\ref{prop-omnis-3}) holds.
\end{proposition}

\begin{remark}\label{rem-productMRQ}
The complete flow $\Phi:\mathbb{R}\times M \rightarrow M$ of the timelike vector field $U$ is an action on $M$ and, if the orbit space has no ancestral pairs, then $Q$ is Hausdorff \cite[Thm. 2.1]{HL01}. So, the principal bundle $\pi:M\rightarrow Q$ has a global section $\sigma:Q\rightarrow M$ \cite[Thm. I.5.7]{KN63} and then $\Sigma=\sigma(Q)\subset M$ is a global slice in $M$ intersecting every integral curve of $U$. The diffeomorphism of (\ref{prop-omnis-3}) in the Proposition \ref{prop-omnis} can be defined by the restriction of the flow $\Phi$ as 
\[
\begin{tabular}{rlcl}
$\tilde{\Phi}:$ & $\mathbb{R}\times Q$ & $\longrightarrow$ & $M$ \\
                & $(t,\gamma)$         & $\longmapsto$     & $\Phi(t,\sigma(\gamma))$
\end{tabular}
\]
where $\gamma$ is the integral curve of $U$ such that $\sigma(\gamma)=\gamma(0)\in \Sigma$.

Since the section $\sigma:Q\rightarrow M$ is an embedding (its inverse is $\left.\pi\right|_{\Sigma}:\Sigma\rightarrow Q$), then $\sigma:Q\rightarrow \Sigma$ is a diffeomorphism. Therefore, we will identify $Q$ with $\Sigma$. So, abusing of the notation, we will write the diffeomorphism $\tilde{\Phi}$ as the flow of $U$, that is
\[
\begin{tabular}{rlcl}
$\Phi:$ & $\mathbb{R}\times Q$ & $\longrightarrow$ & $M$ \\
                & $(t,q)$         & $\longmapsto$     & $\Phi(t,q)$
\end{tabular}
\]
where $q=\gamma(0)=\sigma(\gamma)\in \Sigma$.

Then, in virtue of Proposition \ref{prop-omnis}, we can consider $M=\mathbb{R}\times Q$ where the coordinate $t\in\mathbb{R}$ is the parameter of the flow $\Phi_t$ of the future--directed vector field $U=\partial_t$.
\end{remark}

The past--omniscience of a timelike foliation $\mathcal{F}$ is not sufficient to ensure that the NFOH holds in a strongly causal spacetime with compact leaf space $Q$. Example \ref{ejem-counterexample} contradicts \cite[Prop. 2.1]{CFH22}. The fail in its proof lies in assuming that
\[
\Phi^{-1}_{\Sigma}(\gamma(s_k))=(t_k,x_k)\in [t_{x_{*}},\infty)\times \mathcal{U}_{x_{*}}\subset \mathbb{R}\times Q
\] 
for $\gamma:(a,b)\rightarrow \mathbb{R}\times Q$ a future--directed future--inextendible causal curve and an increasing sequence $\{s_k\}$ such that $s_k \to b$, because it is possible that $(t_k,x_k)\in (-\infty,t_{x_{*}}]$ as Figure \ref{figura-ejemplo-NFOH} shows. 

\begin{example}\label{ejem-counterexample}
Let us consider the portion of $2$--dimensional Minkowski $\mathbb{L}^2$ given by $M=\{(t,x)\in \mathbb{L}^2 :x\in[0,2\pi), t>a\cos x\}$ such that $\vert a\vert >1$ and where any point $(t,x)$ is identified with $(t,x+2\pi)$.

The $t$-coordinate vector field $\partial_t$ defines a timelike foliation $\mathcal{F}$ with a leaf space $Q \simeq \mathbb{S}^1$. 
If we consider the vector field $U=(t-a\cos x)\partial_t$ defining the same foliation $\mathcal{F}$, its flow can be written by 
\begin{equation}\label{eq-flow-counterexample}
\Phi_s(t,x) = \left(a\cdot\cos x +(t-a\cdot\cos x)e^s,x\right) \text{ where } s\in \mathbb{R}
\end{equation}
so $U$ is complete. 
According to Remark \ref{rem-productMRQ}, the hypersurfaces $Q_s=\{s\}\times Q\subset M$ for $s\in\mathbb{R}$ (depending on the global section $\sigma$) can be chosen as $Q_s=\{(t,x)\in M: t=e^s + a\cdot\cos x \}$ and satisfying $\Phi_{r}\left(Q_s\right)=Q_{s+r}$ for any $r\in\mathbb{R}$ (see Figure \ref{figura-ejemplo-NFOH-A}).

Denoting by $S_{\tau}$ the spacelike hypersurface $\{(t,x)\in M: t=\tau\}$, we have that $S_{\tau} \subset I^{-}(t,x)$ for any $x\in [0,2\pi)$ and $t>\tau+2\pi$. 
Hence, for any fixed $x_0$, the leaf of $\mathcal{F}$ parametrized by $\gamma_{x_0}(t)=(t,x_0)$ with $t\in(a\cos x_0,\infty)$ satisfies 
\[
S_{\tau} \subset I^{-}(\gamma_{x_0}(t)) \quad \text{for} \quad t>\tau + 2\pi
\]
for all $\tau \in (-a,\infty)$, then we have $M \subset I^{-}(\gamma_{x_0})$ therefore $I^{-}(\gamma)=M$ for each leaf $\gamma\in \mathcal{F}$. 

Since $\vert a\vert >1$, we can find inextendible (future--directed) timelike curves $\lambda:\mathbb{R}\rightarrow M$ such that $\lim_{s\to +\infty}\lambda(s)=(a\cos x_0,x_0)$ then $\lambda\subset I^{-}(\gamma_{x_0}(t))$ for any $t\in (a\cos x_0,\infty)$ and therefore 
\[
I^{-}(\lambda)\subset I^{-}(\gamma_{x_0}(t))\neq M
\]
concluding that $M$ does not satisfy the NFOH property (see Figure \ref{figura-ejemplo-NFOH-A}).

Observe that $M$ is stably causal because the coordinate $t$ defines a temporal function, i.e. $\mathrm{grad}(t)=-\partial_t$ is timelike past--directed \cite[Thm. 3.56]{MS08}, then $M$ is also strongly causal. But since $M$ is not past--reflecting (see Figure \ref{figura-ejemplo-NFOH-B}), then neither is causally continuous \cite{HS74}, \cite[Def. 3.59]{MS08}.

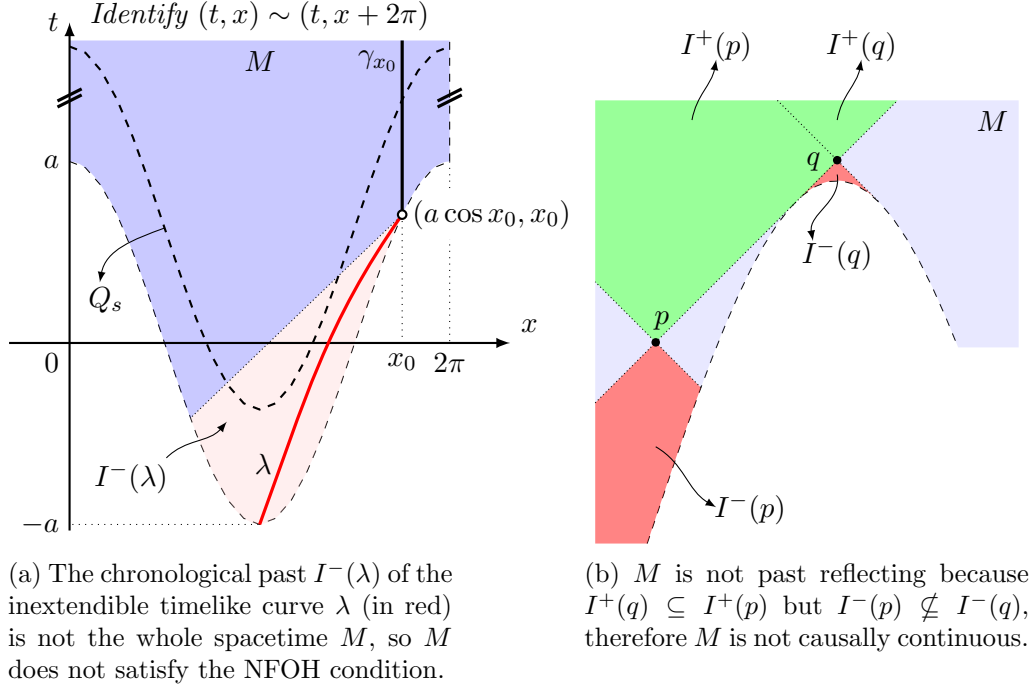
\begin{figure}[h]
  \centering
\begin{subfigure}[t]{0.40\textwidth}
\centering
\begin{tikzpicture}[scale=0.8]
%%%%% Curvas
\draw[dotted] (6.283,3)--(6.283,0) node[anchor=north] {$2\pi$};
\draw[dotted] (3.1215,-3)--(0,-3) node[anchor=east] {$-a$};
\draw (0,3) node[anchor=east] {$a$};
\draw (0,0) node[anchor=north east] {$0$};
\fill (5.4978,2.12135) node[fill=white,anchor=west] {$(a\cos x_0,x_0)$};
\draw[dashed,domain=0:6.283] plot(\x,{3*cos(\x r)});
\fill[domain=1.9978:5.4978,fill=red!20,opacity=0.3] plot(\x,{3*cos(\x r)}); 
\fill[domain=0:1.9978,fill=blue!70,opacity=0.3] plot(\x,{3*cos(\x r)}) -- (5.4978,2.12135) --  plot[domain=5.4978:6.283,fill=blue!20,opacity=0.3](\x,{3*cos(\x r)}) -- (6.283,5)--(0,5)-- cycle;
\draw[dashed] (6.283,3)--(6.283,5);
\draw[very thick] (5.4978,2.12135)  -- (5.4978,5);
\draw[densely dotted] (5.4978,2.12135)  -- (1.9978,-1.2424);
\draw[dashed,thick,domain=0:6.283] plot(\x,{1.9+3*cos(\x r)});
%%%%%% Leyenda
\draw (3.15,5) node[anchor=north] {$M$};
\draw (5.1,5) node[anchor=north] {$\gamma_{x_0}$};
\draw (3.2,-2) node {$\lambda$};
\draw[very thick] (6.083,3.9) -- (6.483,4.1);
\draw[very thick] (6.083,4) -- (6.483,4.2);
\draw[very thick] (-0.2,3.9) -- (0.2,4.1);
\draw[very thick] (-0.2,4) -- (0.2,4.2);
%\draw[-latex,out=-45,in=130] (2.5,0.4) node[anchor=south] {$I^{-}(\lambda)$}  to (3,-0.8);
\draw[-latex,out=45,in=-130] (1.6,-1.9) to (2.6,-1.3);
\draw (1.8,-1.8) node[anchor=north east] {$I^{-}(\lambda)$};
\draw (3.1416,5) node[anchor=south] {Identify $(t,x)\sim (t,x+2\pi)$};
\draw[-latex,out=-160,in=80] (1.57,1.9) to (0.5,1);
\draw (0.6,1.1) node[anchor=north] {$Q_s$};
%%%%% Ejes
\draw[-latex,thick] (-1,0) -- (7.283,0) node[anchor=south west] {$x$};
\draw[-latex,thick] (0,-3.1) -- (0,5.3) node[anchor=east] {$t$};
\draw[red, very thick,out=70,in=-125] (3.1416,-3)  to (5.4978,2.12135);
\draw[dotted] (5.4978,2.12135) -- (5.4978,0) node[anchor=north] {$x_0$};
\filldraw[thick,fill=white] (5.4978,2.12135) circle (2pt);
\end{tikzpicture}
  \caption{The chronological past $I^{-}(\lambda)$ of the inextendible timelike curve $\lambda$ (in red) is not the whole spacetime $M$, so $M$ does not satisfy the NFOH condition.}
  \label{figura-ejemplo-NFOH-A}
  \end{subfigure}
\hspace{0.10\textwidth}
\begin{subfigure}[t]{0.40\textwidth}
\centering
\begin{tikzpicture}[scale=1.6]
%%%%% Curvas
\draw[dashed,domain=-1.5708:1] plot(\x,{3*cos(\x r)});
\fill[domain=-0.34:0.34,fill=red!80,opacity=0.6] plot(\x,{3*cos(\x r)}) -- (0,3.1684) -- cycle; 
\fill[fill=green!80,opacity=0.5] (-1.5,1.6628)  --  (-2,2.1628)-- (-2,3.6628)--(0.5,3.6628) -- cycle;
\fill[fill=blue!30,opacity=0.3] (-2,2.1628)  --  (-2,1.1628)-- (-1.5,1.6628) -- cycle;
\fill[domain=-1.5708:-1.12655,fill=red!80,opacity=0.6]  plot(\x,{3*cos(\x r)}) -- (-1.5,1.6628) -- (-2,1.1628) -- (-2,0) --(-1.5708,0) -- cycle;
\fill[domain=-1.12655:-0.34,fill=blue!30,opacity=0.3]  plot(\x,{3*cos(\x r)}) -- (-1.5,1.6628) -- (-1.12655,1.28935)--cycle; 
\fill[domain=0.34:1,fill=blue!30,opacity=0.3]  plot(\x,{3*cos(\x r)}) -- (1,1.6209) -- (1.5,1.6209) -- (1.5,3.6628) --(0.5,3.6628) -- (0,3.1684) -- cycle; 
\draw[densely dotted] (-2,1.1628)  -- (0.5,3.6628);
\draw[densely dotted] (-1.12655,1.28935) -- (-2,2.1628);
\draw[densely dotted] (-0.4944,3.6628) -- (0.34,2.82826);
%\filldraw[thick,fill=white] (-0.34,2.82826) circle (1pt);
\fill[thick] (-1.5,1.6628) circle (1pt);
\fill[thick] (0,3.1684) circle (1pt);
\draw (-0.05,3.16) node[anchor=east] {$q$};
\draw (-1.45,1.684) node[anchor=south] {$p$};
\draw (1.5,3.6628) node[anchor=north east] {$M$};
%%%%%%%%%%%%%%%%%%%%%%%
\draw[-latex,out=-30,in=150] (-1.5,0.8) to (-1,0.3);
\draw (-1.1,0.5) node[anchor=north west] {$I^{-}(p)$};
%%%%%%%%%%%%%%%%%%%%%%%
\draw[-latex,out=-80,in=100] (0,3.05) to (-0.2,2.5);
\draw (0,2.6) node[anchor=north] {$I^{-}(q)$};
%%%%%%%%%%%%%%%%%%%%%%%
\draw[-latex,out=80,in=-100] (-1.2,3.5) to (-1,4);
\draw (-1,3.9) node[anchor=south] {$I^{+}(p)$};
%%%%%%%%%%%%%%%%%%%%%%%
\draw[-latex,out=80,in=-100] (0,3.5) to (0.2,4);
\draw (0.2,3.9) node[anchor=south] {$I^{+}(q)$};
\end{tikzpicture}
  \caption{$M$ is not past reflecting because $I^{+}(q)\subseteq I^{+}(p)$ but $I^{-}(p) \nsubseteq I^{-}(q)$, therefore $M$ is not causally continuous.}
  \label{figura-ejemplo-NFOH-B}
\end{subfigure}
\caption{The existence of a past--omniscient timelike foliation defined by a complete vector field is not sufficient for $M$ to satisfy the NFOH condition.}
  \label{figura-ejemplo-NFOH}
\end{figure}

On the other hand, notice that, in the case $\vert a \vert \leq 1$, a timelike curve such $\lambda$ cannot exist and the NFOH condition holds on $M$. 

\end{example}

\subsection{Causal decrease}

Recall that the Lie derivative of a tensor $\mathbf{T}\in \mathfrak{T}_k^0(M)$ along a vector field $U\in\mathfrak{X}(M)$ \cite[Prop. 9.21]{On83} is given by
\[
\mathcal{L}_{U}\mathbf{T}(X_1,\ldots,X_k)=\lim_{s\to 0}\frac{1}{s}\left( \mathbf{T}(d\Phi_s(X_1),\ldots ,d\Phi_s(X_k)) - \mathbf{T}(X_1,\ldots,X_k) \right)
\]
where $\Phi_s$ denotes the (local) flow of $U$ and $X_1,\ldots ,X_k\in \mathfrak{X}(M)$.

\begin{definition}\label{def-causaldecreas}
A timelike vector field $U\in\mathfrak{X}(M)$ is said to be \emph{causally decreasing} (resp. \emph{causally increasing}) if 
\[
\mathcal{L}_{U}\mathbf{g}(N,N)\leq 0 \quad \text{(resp. } \mathcal{L}_{U}\mathbf{g}(N,N)\geq 0 \text{)}
\]
for all null vector $N\in T_p M$ and all $p\in M$.
\end{definition}

\begin{definition}\label{def-causaldecreasmap}
A diffeomorphism $f:M\rightarrow M$ is \emph{causally decreasing} if it maps any causal curve to a causal curve and any timelike curve to a timelike curve (preserving causal orientation). 
We will also say that $f$ is \emph{strictly causally decreasing} if any causal curve is mapped to a timelike curve.
\end{definition}

The relation between causal decreasing vector fields and their flows is shown in \cite[Thm. 2.1]{HL01} and stated below. 

\begin{theorem}
Let $U\in\mathfrak{X}(M)$ be a timelike vector field and $\Phi_s$ its flow. Then:
\begin{enumerate}
\item $U$ is causally decreasing if and only if $\Phi_s$ is causally decreasing for $s>0$.
\item If $U$ is strictly causally decreasing then so is $\Phi_s$ for $s>0$.
\end{enumerate}
\end{theorem}

The next theorem \cite[Thm. 2.3]{HL01} is complementary to the previous statement.

\begin{theorem}\label{theo-causal-omnis}
Let $M$ be a chronological spacetime and $U\in\mathfrak{X}(M)$ a causally decreasing and complete future--directed timelike vector field, then the foliation $\mathcal{F}$ of integral curves of $U$ is past--omniscient. 
\end{theorem}

\begin{example}
Observe that, in Example \ref{ejem-counterexample}, using eq. (\ref{eq-flow-counterexample}), we can compute 
\[
d\Phi_s(N)=\begin{pmatrix}
e^s & a(e^s-1)\sin x \\
0 & 1
\end{pmatrix}
\begin{pmatrix}
1 \\
\pm 1
\end{pmatrix}=\begin{pmatrix}
e^s \pm a(e^s-1)\sin x \\
\pm 1
\end{pmatrix}
\]
for the future--directed null vectors $N=(1,\pm 1)$. Then we have 
\[
\mathcal{L}_{U}\mathbf{g}(N,N)=\lim_{s\to 0}\frac{1}{s}\left( \mathbf{g}(d\Phi_s(N),d\Phi_s(N)) - \mathbf{g}(N,N) \right) = -2\left(1\pm a \sin x\right)
\]
and we obtain that 
\[
\mathcal{L}_{U}\mathbf{g}(N,N)\leq 0 \text{ for all }x\in [0,2\pi) \Longleftrightarrow \vert a \vert \leq 1 .
\]
So, the vector field $U$ is causally decreasing whenever the NFOH condition holds in $M$. We wonder how such properties are related in more general spacetimes.
\end{example}

\section{Causally decreasing timelike vector fields and the NFOH condition}\label{sec:causaldecres}

\begin{theorem}[\textbf{Main theorem}]\label{theo-main}
Let $M$ be a distinguishing spacetime. If there exists a causally decreasing and complete future--directed timelike vector field $U\in\mathfrak{X}(M)$ such that its orbit space $Q$ is compact, then there is no future null line in $M$. 
\end{theorem}

In virtue of Remark \ref{rem-productMRQ}, the hypotheses of Theorem \ref{theo-main} permit us to consider $M=\mathbb{R}\times Q$ where the timelike foliation $\mathcal{F}$ is defined by $U=\partial_t\in\mathfrak{X}(\mathbb{R}\times Q)$ whose leaves are given by $\mathcal{F}_q=\mathbb{R}\times \{q\}$ for $q\in Q$.

Before proving the main theorem, we need to determine how inextendible causal curves may behave toward their most extreme future or past according to how they are located in $M=\mathbb{R}\times Q$.

\begin{lemma}\label{lem-infty}
Let $M=\mathbb{R}\times Q$ be a chronological spacetime, $\mathcal{F}$ the foliation defined by the integral curves of the timelike coordinate vector field $U=\partial_t$ of the factor $\mathbb{R}$ and $\lambda = \left( \lambda_1 , \lambda_2 \right):(a,b)\subset\mathbb{R}\rightarrow M=\mathbb{R}\times Q$ be a future--directed causal curve.
\begin{enumerate}[label=\roman*)]
\item If $\lambda$ is past--inextensible and $\mathcal{F}$ is past--omniscient then there is not a decreasing sequence $\{s_n \}_{n=1}^{\infty}\subset (a,b)$ with $\lim_{n\to \infty} s_n =a$ such that $\lim_{n\to \infty}\lambda_1(s_n)=+\infty$ and $\lim_{n\to \infty}\lambda_2(s_n)=q\in Q$.

\item \label{itm:lem-infty-ii} If $\lambda$ is future--inextensible and $\partial_t$ is causally decreasing then there is not an increasing sequence $\{s_n \}_{n=1}^{\infty}\subset (a,b)$ with $\lim_{n\to \infty} s_n =b$ such that $\lim_{n\to \infty}\lambda_1(s_n)=-\infty$ and $\lim_{n\to \infty}\lambda_2(s_n)=q\in Q$.

\item \label{itm:lem-infty-iii} If $M$ is distinguishing, $Q$ compact, $\lambda$ is future--inextensible and $\partial_t$ is causally decreasing then there is not an increasing sequence $\{s_n \}_{n=1}^{\infty}\subset (a,b)$ with $\lim_{n\to \infty} s_n =b$ such that $\lim_{n\to \infty}\lambda(s_n)=(t_0,q_0)\in M$.
\end{enumerate}
\end{lemma}

\begin{proof}

We will denote by $\gamma_q(t)=(t,q)\in \mathbb{R}\times Q$ the integral curve of $U=\partial_t$ corresponding to $q\in Q$. All statements will be proven by \emph{reductio ad absurdum}.

\begin{enumerate}[label=\roman*)]
\item Let us assume the existence of a decreasing sequence $\{s_n \}_{n=1}^{\infty}$ such that $\lim_{n\to \infty}\lambda_1(s_n)=+\infty$ and $\lim_{n\to \infty}\lambda_2(s_n)=q\in Q$.
Since $\mathcal{F}$ is past--omniscient, then for any $s_n$ there exists $t_n\in\mathbb{R}$ such that $\lambda(s_n)\in I^{-}(\gamma_q(t_n))$. 
On the other hand, since $\lambda_2(s_n)\to q$, for any $\tau>t_n$ there exists $s_{\tau}<s_n$ such that $\lambda(s_{\tau})\in I^{+}(\gamma_q(\tau))$. 
Hence, 
\[
\lambda(s_n) \ll \gamma_q(t_n) \ll \gamma_q(\tau) \ll \lambda(s_{\tau})
\]
but $\lambda(s_{\tau}) \leq  \lambda(s_n)$ since $s_{\tau}<s_n$, contradicting that $M$ is chronological.

\item Let us assume the existence of an increasing sequence $\{s_n \}_{n=1}^{\infty}$ such that $\lim_{n\to \infty}\lambda_1(s_n)=-\infty$ and $\lim_{n\to \infty}\lambda_2(s_n)=q\in Q$ where $\lambda$ is causal and future--directed. 
Since $U$ is causally decreasing, then for any $\tau>0$ we have that 
\[
\mu_{\tau}(s):=\Phi_{\tau}\left(\lambda(s)\right)=(\tau+\lambda_1(s) , \lambda_2(s)) 
\]
is also causal and future--directed and it satisfies
\begin{equation}\label{eq-lemma-lambda-infty}
\lim_{n\to \infty}\mu_{\tau}(s_n) = (-\infty,q)  \quad \text{for any $\tau \in \mathbb{R}$}.
\end{equation}
Since $\gamma_q$ is timelike and future--directed, given any $\gamma_q(t_0)$, there exists a neighbourhood $W\subset Q$ of $q$ and $\overline{t}>t_0$ such that $(\overline{t},\infty)\times W \subset I^{+}(\gamma_q(t_0))$. Since $\lambda(s_n)\to (-\infty,q)$ and the limit of eq. (\ref{eq-lemma-lambda-infty}) holds, we can move $\lambda$ by the flow $\Phi_{\tau}$ such that $\mu_{\tau}=\Phi_{\tau}(\lambda)$ intersects $(\overline{t},\infty)\times W \subset I^{+}(\gamma_q(t_0))$, that is, for the given $t_0 \in\mathbb{R}$ there exist $n_0\in\mathbb{N}$ and $\tau>0$ such that $\gamma_q(t_0)\in I^{-}(\mu_{\tau}(s_{n_{0}}))$. 
So, due to $\mu_{\tau}(s_n)$ approaches to $(-\infty,q)$ by eq. (\ref{eq-lemma-lambda-infty}), there exists $n_1>n_0$ such that $\mu_{\tau}(s_{n_1})\in I^{-}\left(\gamma_q(t_0)\right)$. Hence, we get
\[
\gamma_q(t_0) \ll \mu_{\tau}(s_{n_{0}}) \leq \mu_{\tau}(s_{n_1}) \ll  \gamma_q(t_0)
\]
but this contradicts the chronological condition of $M$.  

\item Let us assume that there is an increasing sequence $\{s_n \}_{n=1}^{\infty}\subset (a,b)$ converging to $b$ such that $\lim_{n\to \infty}\lambda(s_n)=(t_0,q_0)\in M$. Since $M$ is distinguishing, any future--inextendible causal curve cannot be totally imprisoned in a compact set \cite[Thm. 2.75]{Mi19}, then $\lambda$ is partially imprisoned in each compact $[-k,k]\times Q$ for all $k\in\mathbb{N}$. By the statement \ref{itm:lem-infty-ii} (taking a subsequence if necessary), there exists a sequence $\{r_n \}_{n=1}^{\infty}\subset (a,b)$ such that $s_n<r_n<s_{n+1}$ for all $n\in \mathbb{N}$ holding $\lim_{n\to \infty}\lambda_1(r_n)=\infty$ and, by compactness of $Q$, $\lim_{n\to \infty}\lambda_2(r_n)=p\in Q$. 
Let us consider any $t_1>t_0$, then $I^{-}(\gamma_{q_0}(t_1))$ is an open neighbourhood of $\gamma_{q_0}(t_0)=(t_0,q_0)$. Moreover, for any $t\in \mathbb{R}$ there exists $n\in\mathbb{N}$ such that $\lambda(r_n)\in I^{+}(\gamma_p(t))$. So, there is $m>n$ such that $\lambda(s_m)\in I^{-}( \gamma_{q_0}(t_1))$, then
\[
\gamma_{p}(t)\ll \lambda(r_n)\leq \lambda(s_m) \ll \gamma_{q_0}(t_1) \quad \text{ for all } t\in \mathbb{R}
\]
and therefore 
\begin{equation}\label{eq-lemma-lambda-p}
I^{-}(\gamma_{p})\subset I^{-}(\gamma_{q_0}(t_1))
\end{equation}
Since $\mathcal{F}$ is past--omniscient, there exists $t\in \mathbb{R}$ such that $\gamma_{q_0}(t_1)\in I^{-}(\gamma_{p}(t))$ but, in virtue of eq. (\ref{eq-lemma-lambda-p}), this contradicts that $M$ is chronological.
\end{enumerate}
\end{proof}

\begin{proof}[\textbf{Proof of the main theorem}]
We will use the notation developed in Lemma \ref{lem-infty}. Given any point $p\in M=\mathbb{R}\times Q$ and any future--inextendible causal curve $\lambda=(\lambda_1,\lambda_2):(a,b)\subset \mathbb{R}\rightarrow M=\mathbb{R}\times Q$ we want to show that $p\in I^{-}(\lambda)$.

By the statements \ref{itm:lem-infty-ii} and \ref{itm:lem-infty-iii} of Lemma \ref{lem-infty}, there exists an increasing sequence $\{s_n\}_{n\in\mathbb{N}}$ such that $\lim_{n\to \infty}\lambda_1(s_n)=+\infty$ and $\lim_{n\to \infty}\lambda_2(s_n)=q\in Q$. 
By Theorem \ref{theo-causal-omnis}, the foliation $\mathcal{F}$ defined by the integral curves of $U=\partial_t$ is past--omniscient, then there exists $t_0\in \mathbb{R}$ such that $p\in I^{-}(\gamma_q(t_0))$, where $\gamma_q$ is the leaf of $\mathcal{F}$ corresponding to $q\in Q$.
Now, since $\lim_{n\to \infty}\lambda(s_n) = (+\infty,q)$, then there exists $t>t_0$ and $n_0\in\mathbb{N}$ such that $\gamma_q(t)\in I^{-}(\lambda(s_{n_0}))$, therefore $p\in I^{-}(\lambda(s_{n_0})) \subset I^{-}(\lambda)$. Finally, the equivalences in Theorem \ref{TIP-NOHequivalence} show our claim.
\end{proof}

By \cite[Props. 2.1, 2.2 \& Thm. 3.1]{BuSa76}, any distinguishing spacetime satisfying the NFOH condition is globally hyperbolic with compact Cauchy surfaces (see also \cite{Ti93}). So, in virtue of Proposition \ref{prop-NFOH-null-lines}, the following corollary is an immediate consequence. 

\begin{corollary}\label{cor-Omega-Cauchy}
Let $M$ be a distinguishing spacetime with a causally decreasing and complete future--directed timelike vector field $U\in\mathfrak{X}(M)$ such that its orbit space $Q$ is compact. Then $M$ is globally hyperbolic with compact Cauchy surfaces. 
\end{corollary}

The previous corollary implies that the assumption that $Q$ is compact is a necessary condition for $M$ to have no null lines.

Recall that the leaf space $Q$ can be immersed in $M$ as hypersurfaces $Q_s=\{s\}\times Q\subset M$ as seen in Remark \ref{rem-productMRQ}, but they may not be Cauchy surfaces. 
In fact, they may have chronologically related points, as suggested by Figure \ref{figura-ejemplo-NFOH-A}. 
Despite this, the leaf space $Q$ is diffeomorphic to the Cauchy surfaces of $M$. Thus, any smooth Cauchy surface $S$ serves as a model for the leaf space $Q$ identifying each leaf $\gamma_q\in Q$ with the singleton $\{z\}=\gamma_q \cap S$ in $S$.

\begin{proposition}\label{prop-Q-diffeo-S}
Under the hypotheses of Theorem \ref{theo-main}, $Q$ is diffeomorphic to any smooth Cauchy surface in $M$
\end{proposition}

\begin{proof}
By Proposition \ref{prop-omnis}, $M$ is diffeomorphic to $\mathbb{R}\times Q$ and, by  \cite[Thm. 3.78]{MS08}, it is also diffeomorphic to $\mathbb{R}\times S$ where $S_{\tau}=\{\tau\}\times S$ are smooth spacelike Cauchy surfaces for any $\tau\in \mathbb{R}$ with $S=S_0$. Then, there exist a diffeomorphism $\psi:\mathbb{R}\times Q \rightarrow \mathbb{R}\times S$ with $\psi(t,q)=(\tau,x)$. 

Since $S\subset M$ is a Cauchy surface, for any $q\in Q$ there exist a unique point $x(q)\in S$ such that $\{x(q)\}=\gamma_q \cap S$ where $\gamma_q$ denotes the integral curve of the vector field $U$ as in Theorem \ref{theo-main}. Assuming that $q_1\neq q_2\in Q$, if $x(q_1)=x(q_2)$ then $\gamma_{q_1}\cap \gamma_{q_2}\neq \varnothing$ but this is not possible because $\gamma_{q_1}$ and $\gamma_{q_2}$ are different integral curves of $U$. Hence the map $x:Q\rightarrow S$ is injective. 

On the other hand, let us denote by $\pi_2:\mathbb{R}\times Q \rightarrow Q$ the canonical projection. So, we define the map $q= \pi_2\circ \left.\psi^{-1}\right|_{S_0}:S\rightarrow Q$. For $x_1\neq x_2\in S$, if $q(x_1)=q(x_2)=q$ then $\gamma_q$ intersects more that once with $S=S_0$, but this contradicts that $S_0$ is a Cauchy surface. Then $q$ is injective. 
Observe that, by construction, $\pi_2$ is the projection along the integral curves of $U$, then $x$ and $q$ must be inverse maps. 

\begin{equation*}
\begin{tikzpicture}[every node/.style={midway}]
\matrix[column sep={6em,between origins},
        row sep={2em}] at (0,0)
{ \node(A)   {$\mathbb{R}\times Q$}  ; & \node(B) {$\mathbb{R}\times S$}; \\
  \node(C) {$Q$}; & \node(D) {$S$} ;                  \\};
\draw[->] (A) -- (B) node[anchor=south]  {$\psi $};
\draw[->] (A) -- (C) node[anchor=east]  {$\pi_2$};
\draw[->] (B) -- (D) node[anchor=west]  {$p_2$};
\draw[->] (C)   -- (D) node[anchor=north] {$x$};
\end{tikzpicture}
\end{equation*}

Finally, if we denote $p_2:\mathbb{R}\times S \rightarrow S$ the canonical projection, then since $\pi_2$ is a submersion and $p_2\circ \psi$ is differentiable then $x$ is differentiable \cite[Prop. 6.1.2]{BC70}. On the other hand, $q$ is differentiable by composition, therefore $x=q^{-1}$ is a diffeomorphism.
\end{proof}

It is possible to state weaker versions of Theorem \ref{theo-main} without the assumption of the existence of the causally decreasing vector field $U$. Observe that the spacetime $M$ in Example \ref{ejem-counterexample} is stably causal but does not hold the NFOH condition. Assuming stronger causality conditions, the conclusions of Lemma \ref{lem-infty} are automatically satisfied, so we can mimic the proof of \cite[Prop. 2.1]{CFH22} or Theorem \ref{theo-main} to show the following result.

\begin{theorem}\label{theo-main-GH}
Let $M$ be a strongly causal spacetime with a past--omniscient timelike foliation $\mathcal{F}$ such that its leaf space $Q$ is compact. If $M$ has no locally naked singularities then the NFOH condition holds in $M$. 
\end{theorem}

\begin{proof}
By Proposition \ref{prop-omnis} and Theorem \ref{theo-causal-omnis}, we can assume that $M=\mathbb{R}\times Q$. Let $\lambda=(\lambda_1,\lambda_2):[0,b)\rightarrow \mathbb{R}\times Q$ be any future--directed future--inextendible causal curve. Since $Q$ is compact, then $\lambda$ cannot be totally or partially imprisoned in $[a,c]\times Q$ for $a<c\in \mathbb{R}$. Then there is an increasing sequence $\{s_n\}$ converging to $b$ such that 
\[
\lim_{n\to +\infty}\lambda_1(s_n)=+\infty \quad \text{ and } \quad \lim_{n\to +\infty}\lambda_2(s_n)=q\in Q  .
\]

Indeed, if $\lim_{n\to +\infty}\lambda_1(s_n)=-\infty$, for any $t\in \mathbb{R}$ there would exist some $n_0\in \mathbb{N}$ such that $\lambda(s_n)\in I^{-}((t,q))$ for all $n>n_0$, then $\lambda\subset I^{-}((t,q))$ and $M$ would have some locally naked singularity.

Now, we can complete this proof following the same arguments used in proof of Theorem \ref{theo-main}.
\end{proof}

Note that, in virtue of the equivalence of strong causality without locally naked singularities and global hyperbolicity \cite[pg. 624]{Pe79}, the Theorem \ref{theo-main-GH} can be enunciated as the following statement.

\begin{corollary}\label{corol-main-NS}
Let $M$ be a globally hyperbolic spacetime. If there exists a past--omniscient timelike foliation $\mathcal{F}$ such that its leaf space $Q$ is compact, then the NFOH condition holds in $M$. 
\end{corollary}

Recall that $(M,\mathbf{g})$ is said to be \emph{conformastationary} if there exists a timelike conformal-Killing vector field $K\in\mathfrak{X}(M)$ such that $\mathcal{L}_{K}\mathbf{g} = \sigma \mathbf{g}$ for some smooth function $\sigma$.

The following corollary arises as an immediate consequence of Theorem \ref{theo-main}. It was stated in \cite[Thm. 1.3]{CFH20} under the hypothesis of global hyperbolicity.

\begin{corollary}
Let $(M,\mathbf{g})$ be distinguishing and conformastationary, such that the timelike conformal--Killing vector field $K$ is complete and the orbit space $Q$ of integral curves of $K$ is compact then $(M,\mathbf{g})$ satisfies the NOH condition.
\end{corollary}

\begin{proof}
The conformal--Killing vector field $K$ holds 
\[
\mathcal{L}_{K}\mathbf{g}(N,N) = \sigma \mathbf{g}(N,N) =0
\]
for every null vector $N$. 
So we have that $(M,\mathbf{g})$ is causally decreasing and causally increasing then, by Theorem \ref{theo-main} and Proposition \ref{TIP-NOHequivalence}, $M$ satisfies the NFOH condition and, dually, also the NPOH condition.
\end{proof}

\section{NFOH condition in GRW spacetimes}\label{sec:GRWs}

Theorem \ref{theo-main} and Proposition \ref{prop-NFOH-null-lines} provide us with sufficient conditions for the NFOH condition to be satisfied, namely, the existence of a complete and causally decreasing future--directed timelike vector field and the compactness of the orbit space $Q$. As an application of the above results, we are going to look for conditions on the warping function of Generalized Robertson--Walker (GRW) spacetimes \cite{Sa98}, \cite{Sa99}. 

Recall that a warped product $M=I\times_f S$ with metric 
\begin{equation}\label{eq-GRW-metric}
\mathbf{g}=-dt^2+f^2(t)\mathbf{h}
\end{equation}
is a GRW spacetime if $I\subset \mathbb{R}$ is an open interval, $\mathbf{h}$ is a Riemannian metric in $S$ and $f:I\rightarrow (0,\infty)$ is the warping function or the scale factor. Observe that, $\mathcal{F}=\{ I\times\{q\}:q\in S\}$ is a timelike foliation defined by the vector field $\partial_t$. 

We will consider three examples depending on the type of the chosen interval $I=\mathbb{R}$, $I=(0,\infty)$ or $I=(0,b)$ for $0<b<\infty$.

In what follows, we will identify vector fields in $\mathfrak{X}(S)$ with their lifts $\mathfrak{L}(S)$ to $M=\mathbb{R}\times S$ \cite[Sect. 1]{On83}.

\begin{example}
Consider the GRW spacetime $M=\mathbb{R}\times S$ equipped with the metric $\mathbf{g}=-dt^2+f^2(t)\mathbf{h}$. The vector field $U=\partial_t$ is complete since its flow is given by $\Phi_s(t,p)=(t+s,p)$ with $s\in\mathbb{R}$. 
Then we have 
\[
d\Phi_s(\partial_t)=\partial_t \quad \text{ and } \quad d\Phi_s(X)=X \text{ for any } X\in\mathfrak{X}(S) .
\]
If $X\in\mathfrak{X}(S)$ holds that $\mathbf{h}(X,X)=1$, then $N= f(t)\partial_t + X$ is a null vector field. So, we have 
\[
d\Phi_s(N)=f(t)\partial_t + X
\]
and therefore
\begin{align*}
\mathcal{L}_{U}\mathbf{g}(N,N) & =\lim_{s\to 0}\frac{1}{s}\mathbf{g}\left( d\Phi_s(N), d\Phi_s(N)\right) = \lim_{s\to 0}\frac{1}{s}\left(-f^2(t)+f^2(t+s)\mathbf{h}(X,X) \right)= \\
&=\lim_{s\to 0}\frac{1}{s}\left(-f^2(t)+f^2(t+s) \right)= 2f(t)f'(t)
\end{align*}

Since $f>0$, then 
\begin{equation}\label{ex-R}
U=\partial_t \text{ is causally decreasing }\Leftrightarrow ~f'(t)\leq 0 \text{ for all }t\in\mathbb{R}.
\end{equation}
Hence, if $f'\leq 0$ then $M$ satisfies the NFOH condition.
\end{example}

\begin{example}
Now, consider $M=(0,\infty)\times S$ equipped with the metric $\mathbf{g}$ of eq. (\ref{eq-GRW-metric}). The vector field $U=t\partial_t$ is complete and its flow is given by 
\[
\Phi_s(t,p)=(te^s,p) \text{ with } s\in\mathbb{R}. 
\]
Then we have 
\[
d\Phi_s(\partial_t)=e^s\partial_t \quad \text{ and } \quad d\Phi_s(X)=X \text{ for any } X\in\mathfrak{X}(S) .
\]
Again, if $X\in\mathfrak{X}(S)$ is such that $\mathbf{h}(X,X)=1$, then $N= f(t)\partial_t + X$ is a null vector field. So, we have 
\[
d\Phi_s(N)=f(t)e^s\partial_t + X
\]
and then, by a straightforward computing using, for example the L'H\^{o}pital rule, we get
\begin{align*}
\mathcal{L}_{U}\mathbf{g}(N,N) & =\lim_{s\to 0}\frac{1}{s}\left(-f^2(t)e^{2s}+f^2(te^s) \right)= 2f(t)\left(-f(t)+tf'(t)\right)
\end{align*}

Since $f>0$, then 
\[
U=t\partial_t \text{ is causally decreasing }\Leftrightarrow ~-f(t)+tf'(t)\leq 0 \text{ for all }t\in (0,\infty)
\]
or equivalently
\begin{equation}\label{ex-0inf}
U=t\partial_t \text{ is causally decreasing }\Leftrightarrow ~\frac{f'(t)}{f(t)}\leq \frac{1}{t} \text{ for all }t\in (0,\infty).
\end{equation}

For example, let us consider that $M$ has an initial big bang, then $\lim_{t\to 0}f(t)=0$. Then for $f(t)=t^{\alpha}\phi(t)$ for $\alpha>0$ and $\phi>0$ with $\lim_{t\to 0}\phi(t)=a\in(0,\infty)$, by condition (\ref{ex-0inf}), we have that
\[
\frac{f'(t)}{f(t)}=\frac{\alpha t^{\alpha -1}\phi(t)+t^{\alpha}\phi'(t)}{t^{\alpha}\phi(t)}= \frac{\alpha}{t}+\frac{\phi'(t)}{\phi(t)}\leq  \frac{1}{t} \Rightarrow \frac{\phi'(t)}{\phi(t)}\leq  \frac{1-\alpha}{t}
\] 
Observe that if $\alpha=1$ then $\phi$ must be non--increasing and if $\alpha\in(0,1)$ then the factor $\phi$ could be increasing but satisfying $\lim_{t\to \infty}\phi'(t)=0$.  
So, some possible scale factors satisfying the condition (\ref{ex-0inf}) are 
\[
\left\{
\begin{tabular}{lrl}
$f(t)=t^{\alpha}$ & for & $\alpha\in (0,1]$ \\
$f(t)=t^{\alpha}(a+t^{\beta})$ & for & $\alpha\in (0,1)$, $\alpha + \beta<1$ and $a>0$\\
$f(t)=t^{\alpha}e^{k-\beta t}$ & for & $\alpha\in (0,1]$, $\beta>0$ and $k\in\mathbb{R}$\\
\end{tabular}
\right.
\]
\end{example}

\begin{example}
Consider $M=(0,b)\times S$ with $0<b<\infty$ with the metric $\mathbf{g}$ of eq. (\ref{eq-GRW-metric}). Now, the complete vector field is $U=t(b-t)\partial_t$. The flow of $U$ can be written by 
\[
\Phi_s(t,p)=\left(\frac{bt}{t+(b-t)e^{-bs}},p\right) \text{ with } s\in\mathbb{R}. 
\]
Then we have 
\[
d\Phi_s(\partial_t)=\frac{b^2 e^{-bs}}{(t+(b-t)e^{-bs})^2}\partial_t \quad \text{ and } \quad d\Phi_s(X)=X \text{ for any } X\in\mathfrak{X}(S) .
\]
If we take the null vector field $N= f(t)\partial_t + X$ with $X\in\mathfrak{X}(S)$ such that $\mathbf{h}(X,X)=1$, then we have 
\[
d\Phi_s(N)=\frac{f(t)b^2 e^{-bs}}{(t+(b-t)e^{-bs})^2}\partial_t + X
\]
and hence
\begin{align*}
\mathcal{L}_{U}\mathbf{g}(N,N) & =\lim_{s\to 0}\frac{1}{s}\left(-\frac{f^2(t)b^4 e^{-2bs}}{(t+(b-t)e^{-bs})^4}        +f^2\left(\frac{bt}{t+(b-t)e^{-bs}}\right) \right)= \\
&=  2f(t)\left(-(b-2t)f(t)+t(b-t)f'(t)\right)
\end{align*}
Therefore 
\begin{equation}\label{ex-0b}
U=t(b-t)\partial_t \text{ is causally decreasing }\Leftrightarrow ~\frac{f'(t)}{f(t)}\leq \frac{b-2t}{t(b-t)} \text{ for all }t\in(0,b) .
\end{equation}

If $M$ represents a cosmological model with initial big bang and final big crunch, then we can choose $f(t)=t^{\alpha}(b-t)^{\beta}$ for $\alpha,\beta >0$. So, by condition (\ref{ex-0b}), for $t\in (0,b)$, we have that
\[
\frac{f'(t)}{f(t)}=\frac{\alpha b -(\alpha+\beta)t}{t(b-t)}\leq \frac{b-2t}{t(b-t)}  \Rightarrow \alpha\leq 1 \text{ and } \beta\geq 1
\] 
Observe that the power $\beta$ must not be lesser that the power $\alpha$ for such function $f$.  

If we consider $f(t)=t^{\alpha}(b-t)^{\beta}\phi(t)$ for $\alpha,\beta >0$ and $\phi>0$ with $\lim_{t\to 0}\phi(t)\in (0,\infty)$ and $\lim_{t\to b}\phi(t)\in (0,\infty)$, by condition (\ref{ex-0b}), we have that
\[
\frac{f'(t)}{f(t)}=\frac{\alpha b -(\alpha+\beta)t}{t(b-t)}+\frac{\phi'(t)}{\phi(t)}\leq \frac{b-2t}{t(b-t)}   
\]
whence
\begin{equation}\label{eq-0b-phi}
\frac{\phi'(t)}{\phi(t)}\leq  \frac{b(1-\alpha)+(\alpha + \beta -2)t}{t(b-t)} 
\end{equation}
In this case, there are more parameters involved and the calculations are not so easy to visualise. For example, if we consider $\phi(t)=(d_1+t)^k (b+d_2-t)^m$ for $d_1 , d_2 >0$ then, inequality (\ref{eq-0b-phi}) results 
\[
 \frac{\phi'(t)}{\phi(t)}= \frac{k(b+d_2)-md_1 -(k+m)t}{(d_1+t)(b+d_2-t)} \leq  \frac{b(1-\alpha)+(\alpha + \beta -2)t}{t(b-t)}
\]
and it holds if 
\[
k(b+d_2)-md_1\leq b(1-\alpha) \quad \text{ and } \quad \alpha+\beta+k+m\geq 2  .
\]
\end{example}

Notice that $H(t)=\frac{f'(t)}{f(t)}$ is the mean curvature of $S_t=\{t\}\times S$ (by \cite[Prop. 7.35]{On83}), then conditions (\ref{ex-R}), (\ref{ex-0inf}) and (\ref{ex-0b}) refer to that geometric feature. They state an upper bound function on the mean curvature of $S_t$ for $t\in I\subset \mathbb{R}$. 

Observe that condition (\ref{ex-R}) implies that the scale factor must be non--increasing for all $t\in \mathbb{R}$. Moreover, Cauchy surfaces $S_t$ could be maximal, i.e. $H(t)=0$.

The condition (\ref{ex-0inf}) allows for both the existence of maximal Cauchy surfaces and their non--existence. It is also remarkable that the scale factor $f$ can be increasing since, for any $t\in (0,\infty)$, there is always room for the mean curvature $H(t)$ to be positive (see Figure \ref{figura-ejemplo-GRW-A}).

In case of condition (\ref{ex-0b}), the scale factor $f$ could be increasing until half-time $t=b/2$ and thereafter it has to be decreasing in such a way that $\lim_{t\to b}H(t)=-\infty$, then Cauchy surfaces $S_t$ only can be maximal if $t\in(0,b/2]$ (see Figure \ref{figura-ejemplo-GRW-B}).

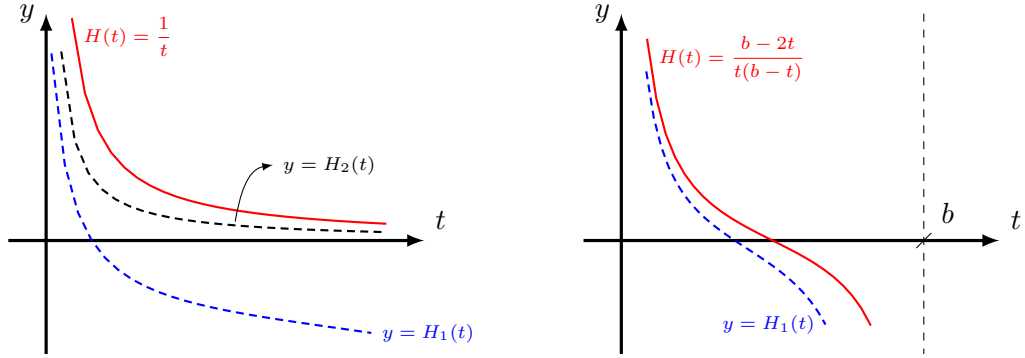
\begin{figure}[h]
  \centering
\begin{subfigure}[t]{0.40\textwidth}
\centering
\begin{tikzpicture}[scale=1]
%%%%% Ejes
\draw[-latex,very thick] (-0.5,0) -- (5,0) node[anchor=south west] {$t$};
\draw[-latex,very thick] (0,-1.5) -- (0,3) node[anchor=east] {$y$};
%%%%% Curvas
\draw[domain=0.34:4.5,thick, red] plot(\x,{1/(\x)}); 
\draw[domain=0.2:4.5,densely dashed,thick] plot(\x,{1/((2*\x))});
\draw[domain=0.27:4.5,densely dashed,thick, blue] plot(-0.2+\x,{-1+(9-3*\x)/((9-\x)*(\x))}) node[anchor=west] {\tiny{$y=H_1(t)$}};
%\draw[domain=0.5:4.5,densely dotted,thick] plot(\x,{1/((0.5+\x)^2)});
%\draw[domain=1:2.8,densely dotted,thick] plot({0.25 + 1/(2*(\x))^2},{\x});
\draw[red] (0.36,2.7)  node[anchor=west] {\tiny{$H(t)=\displaystyle{\frac{1}{t}}$}};
\draw[-latex,out=80,in=190] (2.5,0.25) to (3,1);
\draw (3,1)  node[anchor=west] {\tiny{$y=H_2(t)$}};
\end{tikzpicture}
  \caption{A scale factor $f$ corresponding to $H_1$ implies the existence of a maximal Cauchy surface $S_t$. In the case of $H_2$, there is no maximal Cauchy surface $S_t$ but it tends to be maximal when $t\to \infty$. In both cases, $M$ would satisfy the NFOH condition.}
  \label{figura-ejemplo-GRW-A}
  \end{subfigure}
\hspace{0.10\textwidth}
\begin{subfigure}[t]{0.40\textwidth}
\centering
\begin{tikzpicture}[scale=1]
%%%%% Ejes
\draw[-latex,very thick] (-0.5,0) -- (5,0) node[anchor=south west] {$t$};
\draw[-latex,very thick] (0,-1.5) -- (0,3) node[anchor=east] {$y$};
%%%%% Curvas
\draw[domain=0.34:3.3,thick, red] plot(\x,{(4-2*(\x))/((\x)*(4-\x))}); 
\draw[domain=0.33:2.7,densely dashed,thick, blue] plot(\x,{(3-2*(\x))/((\x)*(3.5-\x))}) node[anchor=east] {\tiny{$y=H_1(t)$}};
%\draw[domain=0:4.5,densely dotted,thick] plot(\x,{(\x)/(1+(\x)^2)}) ;
\draw[dashed] (4,3) -- (4,-1.5);
\draw (3.9,-0.1) -- (4.1,0.1) node[anchor=south west] {$b$};
\draw[red] (0.36,2.4)  node[anchor=west] {\tiny{$H(t)=\displaystyle{\frac{b-2t}{t(b-t)}}$}};
\end{tikzpicture}
  \caption{Any scale factor $f$ holding condition (\ref{ex-0b}) implies the existence of a maximal Cauchy surface $S_t$.}
  \label{figura-ejemplo-GRW-B}
\end{subfigure}
\caption{The condition for the complete vector field $U=a\partial_t$ to be causally decreasing in a GRW spacetime is a condition on the mean curvature of the Cauchy surfaces $S_t$.}
  \label{figura-ejemplo-GRW}
\end{figure}

It is also important to note that conditions (\ref{ex-R}), (\ref{ex-0inf}) and (\ref{ex-0b}) obtained in the previous examples imply that the vector field $\partial_t$ is causally decreasing and, as a consequence, $(M,\mathbf{g})$ satisfies the NFOH condition. 
The dual conditions for $\partial_t$ to be causally increasing follow from applying the reverse inequality; thus, $\partial_t$ will be \emph{causally constant} (increasing and decreasing) when 
\[
\left\{
\begin{tabular}{lcll}
$H(t)=0$ & $\Rightarrow$ & $f(t)=C$ constant & for $M=\mathbb{R}\times S$ \\
$H(t)=t^{-1}$ & $\Rightarrow$ & $f(t)=t$ & for $M=(0,\infty)\times S$ \\
$H(t)=\frac{b-2t}{t(b-t)}$ & $\Rightarrow$ & $f(t)=t(b-t)$ & for $M=(0,b)\times S$
\end{tabular}
\right.
\]
In these cases, the condition NOH holds in $(M,\mathbf{g})$. 
Furthermore, one might think that the causal decrease of $\partial_t$, at least in a weak sense, is necessary for the NFOH condition to hold; for example, if there are no future null lines, then $\partial_t$ is causally decreasing in a neighbourhood of the future limit of spacetime (i.e. $(c,\infty)$ or $(c,b)$, depending on the model), however, this is far from being the case. The following example shows that $(M,\mathbf{g})$ can be causally decreasing at its beginning and causally increasing at its end whilst satisfying the NOH condition, that is, NFOH and NPOH conditions simultaneously.

\begin{example}

Let us consider $M=\mathbb{R}\times \mathbb{S}^1$ where $\mathbb{S}^1$ is the 1--sphere. We will take the metric
\[
\mathbf{g}=-dt^2+(1+t^2)d\theta^2
\]
where $t\in \mathbb{R}$ and $\theta\in [0,2\pi)\simeq\mathbb{S}^1$. 
According to (\ref{eq-GRW-metric}), we have that $f(t)=\sqrt{1+t^2}$ and then
\[
f'(t)=\frac{t}{\sqrt{1+t^2}},\quad H(t)=\frac{t}{1+t^2}
\]
whence $f$ is causally decreasing in $(-\infty,0)$ and increasing in $(0,\infty)$.

Let us see that there are no null lines in $(M,\mathbf{g})$. Fix any point $p=(t_0,\theta_0)\in M$ and any future--inextendible causal curve $\lambda:[0,1)\rightarrow M$. Since $\lambda$ is future--inextendible, there exist $\overline{t}\in \mathbb{R}$ such that $\lambda\cap (\{t\}\times\mathbb{S}^1)\neq\varnothing$ for all $t\geq \overline{t}$, then we can choose a point $q=(t_1,\theta_1)\in \mathrm{Im}(\lambda)$ such that $t_1>t_0$ and, since the function $h(s)=\vert s+\sqrt{1+s^2} \vert$ is such that
\[
\lim_{s\to -\infty}h(s)=0 \quad \text{and} \quad \lim_{s\to +\infty}h(s)=+\infty 
\]
then $t_1$ can hold
\[
K:=\frac{\vert t_1+\sqrt{1+t_1^2}\vert}{\vert t_0+\sqrt{1+t_0^2}\vert} > e^{\pi}  .
\]
Moreover, since the angle between two points of $\mathbb{S}^1$ is lesser or equal that $\pi$, then we can assume, without any lack of generality, that $0\leq \theta_0 \leq \theta_1 < 2\pi$ with $\theta_1-\theta_0\leq \pi$. Now, take $\alpha=\frac{\theta_1 - \theta_0}{\ln(K)}$ and  $c=t_1-t_0$, then the curve $\mu:[0,c]\rightarrow M$ defined by 
\[
\mu(s)=\left( t_0 + s  ~, \theta_0 + \alpha \cdot\ln\left( \frac{\vert t_0+s+\sqrt{1+(t_0+s)^2}\vert}{\vert t_0+\sqrt{1+t_0^2}\vert} \right) \right)
\]
joins $\mu(0)=p=(t_0,\theta_0)$ with $\mu(c)=q=(t_1,\theta_1)$. Observe that, since $\vert \alpha\vert < 1$ and $\mu'(s)=\left(1,\frac{\alpha}{\sqrt{1+(t_0+s)^2}}\right)$, so  
\[
\mathbf{g}(\mu',\mu')=-1+(1+(t_0+s)^2)\frac{\alpha^2}{1+(t_0+s)^2}=-1+\alpha^2  < 0
\]
and then $\mu$ is timelike. Hence, we have 
\[
p\in I^{-}(\mu(c))=I^{-}(q)\subset I^{-}(\lambda)\quad \Rightarrow \quad I^{-}(\lambda)=M  .
\]
and there are no future null lines in $(M,\mathbf{g})$. Due to symmetry with respect to the $t$--coordinate, there are no past null lines either. Therefore $(M,\mathbf{g})$ satisfies the NOH condition.

\end{example}

\section{Conclusion}

In this work, Example \ref{ejem-counterexample} demonstrates the existence of a strongly causal spacetime $M$ equipped with a complete timelike vector field $U \in \mathfrak{X}(M)$ whose orbit space $Q$ is compact and whose integral curves are past--omniscient, yet fails to satisfy the NFOH condition and consequently possesses future null lines. This result serves as a counterexample to \cite[Prop. 2.1]{CFH22}, indicating that additional hypotheses are required to recover the original result.

In Theorem \ref{theo-main}, we establish that by assuming $U$ is causally decreasing, the causality requirement on $M$ can be relaxed to the level of a distinguishing spacetime.

Without the assumption of $U$ being causally decreasing, it becomes necessary to impose causality conditions stronger than stable causality to ensure the non--existence of future null lines. This case is addressed in Theorem \ref{theo-main-GH} and Corollary \ref{corol-main-NS} for globally hyperbolic spacetimes. Whether these results hold for intermediate steps of the causal ladder, such as causal continuity or causal simplicity, remains an open question.

Finally, as an application of the main theorem, we investigate the causal decrease of the coordinate timelike vector field $\partial_t$ in Generalized Robertson--Walker (GRW) spacetimes $M = I \times_f S$ with metric (\ref{eq-GRW-metric}). For the cases $I = \mathbb{R}$, $I = (0, \infty)$, and $I = (0, b)$, we derive specific conditions on the warping function $f$ and, consequently, also on the mean curvature of the Cauchy surfaces $S_t = \{t\} \times S$ to ensure the causal decrease of $\partial_t$. Thus, by Theorem \ref{theo-main}, these conditions also guarantee that there are no future null lines in those GRW spacetimes.

\section*{Acknowledgements}

The author expresses sincere gratitude to Professor Javier Lafuente. Although the author’s interest in causality theory predates their collaboration, Professor Lafuente’s early guidance has been a major influence on the ideas developed in this work.

The author also thanks the two anonymous reviewers for their constructive comments and valuable feedback, which have significantly improved the article.

%%%%%%%%%%%%%%%%%%%%%%%%%%%%%%%%%%%%%%%%%%%%%%%%%%%%%

\end{document}